\documentclass[copyright,creativecommons,noderivs,noncommercial]{eptcs}
\providecommand{\event}{B. Klin, P. Soboci\'nski (Eds.): \\
6th Workshop on Structural Operational Semantics (SOS'09)}
\providecommand{\volume}{19}
\providecommand{\anno}{2010}
\providecommand{\firstpage}{92}
\providecommand{\eid}{7}

\usepackage{amsmath}
\usepackage{amssymb}
\usepackage{amsthm}
\usepackage{graphicx}
\usepackage{xspace,stmaryrd}
\DeclareSymbolFont{letters}{OML}{txmi}{m}{it} 

\usepackage{tikz}
\usetikzlibrary{arrows,automata,shapes,snakes}

\newtheorem{defn}{Definition}[section]
\newtheorem{lem}[defn]{Lemma}
\newtheorem{thm}[defn]{Theorem}

\newtheorem{exmp}[defn]{Example}

\newcommand{\Nat}{\textit{I\hspace{-.5ex}N}}

\newcommand{\sosrule}[2]{\frac{\raisebox{.7ex}{\normalsize{$#1$}}}
                        {\raisebox{-1.0ex}{\normalsize{$#2$}}}}

\newcommand{\bisim}{\;\underline{\hspace*{-0.15ex}
                        \leftrightarrow\hspace*{-0.15ex}}\;}

\def\lparal{\mathbin{\setbox0=\hbox{$\|$}%
        \dimen0=\dp0 \advance\dimen0 -1.5pt \dp0=\dimen0%
        \underline{\kern-1.5pt\box0\kern1.5pt}}}

\newcommand{\rel}{\mathcal{R}}

\newcommand{\contents}[1]{\noindent$\blacktriangleright$ #1 \\}

\newcommand{\ie}{\textit{i.e.}\xspace }
\newcommand{\eg}{\textit{e.g.}\xspace }
\newcommand{\viz}{\textit{viz.}\xspace }

\newcommand{\ibid}{\textit{ibid.}\xspace }

\newcommand{\isdef}{=}
\newcommand{\ap}{{:}}

\renewcommand{\le}{\mathrel{\leqslant}}

\newcommand{\tle}{\ensuremath{\mathrel{\trianglelefteqslant}}}

\newcommand{\bnd}[1]{\ensuremath{\mathsf{bnd}(#1)}}

\newcommand{\occ}[1]{\ensuremath{\mathsf{occ}(#1)}}
\newcommand{\normalise}[1]{\ensuremath{\mathsf{norm}(#1)}}
\newcommand{\reduce}[1]{\ensuremath{\mathsf{reduce}(#1)}}

\newcommand{\bool}{\ensuremath{\mathbb{B}}}
\newcommand{\true}{\ensuremath{\mathsf{true}}}
\newcommand{\false}{\ensuremath{\mathsf{false}}}

\newcommand{\mc}[1]{\ensuremath{\mathcal{#1}}}

\newcommand{\rank}[1]{\ensuremath{\mathsf{rank}(#1)}}
\newcommand{\cad}[3]{\ensuremath{\mathsf{rank}_{#1,#2}(#3)}}
\newcommand{\ah}[1]{\ensuremath{\mathsf{ah}(#1)}}

\newcommand{\sem}[1]{\ensuremath{[\![ #1 ]\!]}}

\newcommand{\up}{\blacktriangle}
\newcommand{\down}{\blacktriangledown}

\title{Analysis of Boolean Equation Systems\\ through Structure Graphs}

\author{
Michel A. Reniers
\quad Tim A.C. Willemse
\institute{Department of Computer Science, Eindhoven University of Technology (TU/e),
\\ {P.O.~Box~513}, NL-5600~MB~~Eindhoven, The Netherlands
}
}

\def\authorrunning{Reniers \& Willemse}
\def\titlerunning{Analysis of Boolean Equation Systems through Structure Graphs}

\begin{document}
\maketitle

\bibliographystyle{plain}

\begin{abstract}
We analyse the problem of solving Boolean equation systems through the use
of \emph{structure graphs}. The latter are obtained through an elegant
set of Plotkin-style deduction rules.  Our main contribution is that
we show that equation systems with bisimilar structure graphs have the
same solution.  We show that our work conservatively extends earlier work,
conducted by Keiren and Willemse, in which \emph{dependency graphs} were
used to analyse a subclass of Boolean equation systems, \viz, equation
systems in \emph{standard recursive form}. We illustrate our approach by
a small example, demonstrating the effect of simplifying an equation system
through minimisation of its structure graph.

\end{abstract}

\section{Introduction}\label{Sect:intro}

\emph{Boolean equation systems (BESs)}~\cite{Lar:92,Mad:97} essentially
consist of sequences of fixed-point equations in the Boolean lattice.
Their merit is in their use for solving a variety of advanced
verification problems in a uniform manner, \viz, by solving the equation
system itself; such problems include local and global model checking
problems, see \eg~\cite{Mad:97} and equivalence checking problems,
see~\cite{Mat:03,CPPW:07}. Through dedicated encodings that act on a
combination of, \eg, labeled transition systems and temporal formulae,
equation systems encoding a particular verification problem can be
obtained efficiently, \ie, in polynomial time.  The size of the resulting
equation system is dependent on the input and the verification problem:
the $\mu$-calculus model checking problem, for instance, yields equation
systems of size $\mc{O}(n f)$, where $n$ is the size of the state space
and $f$ the size of the modal formula. As a result, equation systems
can suffer from a phenomenon akin to the state explosion problem.

Solving an equation system is known to be a computationally hard problem:
it is in $\text{NP} \cap \text{co-NP}$, see, \eg~\cite{Mad:97}; in fact,
Jurdzi\'nski showed that it is in $\text{UP} \cap \text{co-UP}$, see~\cite{Jur:98}.
Currently, the most efficient algorithm for solving equation
systems (at least from a theoretical stance), is the \emph{bigstep}
algorithm\footnote{Technically, this algorithm is used to compute
the set of winning states for a player in a \emph{Parity Game},
but this problem is equivalent to the problem of solving an equation
system.} due to Schewe~\cite{Schewe2007}. This algorithm has run-time
complexity $\mc{O}(n\ m^{\ah{\mc{E}}/3})$, where $n$ corresponds to
the number of equations in an equation system $\mc{E}$, $m$ to the
cumulative size of the right-hand sides of these, and $\ah{\mc{E}}$
to the number of alternations of fixed-point signs in the equation
system. This run-time complexity provides a practical motivation for
investigating methods for efficiently reducing the size of these
parameters.  In the absence of notions such as a behaviour of an
equation system, an unorthodox strategy in this setting is the use of
bisimulation-inspired minimisation techniques.  Nevertheless, recent work
by Keiren and Willemse~\cite{KeirenWillemse2009a} demonstrates that two
such minimisations are not only theoretically but also practically very
cost-effective:  they yield massive reductions of the size of equation
systems, they do not come with memory penalties, and the time required
for solving the original equation system significantly exceeds the
time required for minimisation and subsequent solving of the minimised
equation system.

In \ibid, the minimisations are only obtained for a strict subclass of
equation systems, \viz, equation systems in \emph{standard recursive form
(SRF)}. The minimisation technique relies on bisimulation minimisations of
\emph{dependency graphs}~\cite{Kei:06,KeirenWillemse2009a} underlying the
equation systems in SRF.  These graphs basically reflect the (possibly
mutual) dependencies of the equations in an equation system in SRF.
While from a practical viewpoint, the class of equation systems in
SRF does not pose any limitations to the applicability of the method
(every equation system can be brought into SRF without changing the
solution to the proposition variables of the original equation system),
the transformation comes at the cost of a blow-up in size. While
this blow-up is only polynomial in size, its effects on the minimising
capabilities were thus far not clear. As a result of our developed theory
we are able to show that the reduction to SRF does not adversely affect
the minimising capabilities of strong bisimulation. This follows from
the fact that bisimilarity on structure graphs is a congruence for
\emph{normalisation}, \ie, an operation that transforms an equation
system into SRF.  More importantly, the required transformation into
SRF complicates the development of meta-theory for equation systems.
For instance, it hinders addressing questions such as whether the
minimisation of equation systems is always favourable over minimising
input specifications prior to encoding the problem as equation systems.

The main problem in generalising the results that are obtained through
the analysis of dependency graphs is that it is hard to elegantly capture
the structure of an equation system, without resulting in a parse-tree of
the equation system. In addition, the arbitrary nesting levels of Boolean
operators in equation systems complicates a straightforward definition of
bisimilarity for such general equation systems. We solve these issues
by using a set of deduction rules in Plotkin style~\cite{Plotkin04a}
to map the equation systems onto \emph{structure graphs}. The latter
generalise dependency graphs by dropping the requirement that each node
necessarily represents a proposition variable occurring at the left-hand
side of some equation and adding facilities for reasoning about Boolean
constants $\true$ and $\false$. Motivated by computational complexity,
in defining our deduction rules, we necessarily must leverage between
\emph{simplicity} and \emph{coarseness}. This is achieved by choosing to
support only rules of commutativity and associativity of the Boolean
operators, and not, \eg, distributivity and absorption rules. The
rationale behind this choice is that commutativity and associativity,
which are hard-coded in equation systems in SRF (and therefore in their
underlying dependency graphs) have proven to be sufficiently powerful
for obtaining reductions from an arbitrary number of equations to a single equation.

\paragraph{Related Work.}
Various types of graphs for Boolean equation systems have appeared in the
literature. In~\cite{Mad:97}, Mader considers dependency graphs consisting
of nodes representing equations and edges representing the fact that
one equation depends on the value of another equation. The structure
of the right-hand sides of the equations can in no sense be captured by
these graphs. Kein\"anen~\cite{Kei:06} extends the dependency graphs of
Mader by decorating the nodes with at most one of the Boolean operators
$\wedge$ and $\vee$, and, in addition, a natural number that abstractly
represents the fixed-point sign of the equation.  However, the dependency
graphs of \ibid, only allow for capturing equation systems in SRF. Keiren
and Willemse~\cite{KeirenWillemse2009a} use these dependency graphs to
investigate two notions of bisimulation, \viz, \emph{strong bisimulation},
and a weakened variation thereof, called \emph{idempotence-identifying bisimulation},
and their theoretical and practical use for minimising equation systems.
The dependency graphs of~\cite{Kei:06,KeirenWillemse2009a}, in turn,
are closely related to \emph{Parity Games}, in which players aim to
win an infinite game. It has been shown that the latter problem is
equivalent to solving a Boolean equation system.  Simulation relations for Parity Games have been studied in, among others~\cite{FW:06}. Finally, we mention the framework of \emph{Switching
Graphs}~\cite{GP:09}, which have two kinds of edges: ordinary edges and
\emph{switches}, which can be set to one of two destinations. Switching
Graphs are more general than dependency graphs, but are still inadequate for
directly capturing the structure of the entire class of Boolean equation
systems. Note that in this setting, the \emph{$v$-parity loop problem}
is equivalent to the problem of solving Boolean equation systems.

\paragraph{Outline.} For completeness, we provide a brief
overview of the formal framework of Boolean equation systems in
Section~\ref{Sect:preliminaries}.  Section~\ref{Sect:Structure_Graphs}
subsequently introduces the concept of a structure
graph and presents deduction rules for generating these
from an equation system. Our main results are presented in
Section~\ref{Sect:Preservation}, and an application thereof can be
found in Section~\ref{Sect:Application}. Section~\ref{Sect:Conclusions}
finishes with concluding remarks.

\section{Preliminaries}\label{Sect:preliminaries}

A Boolean equation system is a finite sequence of least and greatest
fixed point equations, where each right-hand side of an equation is a
proposition in positive form. For an in-depth treatment of the
associated theory for model checking the modal $\mu$-calculus,
we refer to~\cite{Mad:97}. In the remainder of this section, we outline
only the theory that is required for understanding the results obtained
in this paper.

\begin{defn}
A \emph{Boolean equation system (BES)} $\mc{E}$ is defined by the following grammar:
\[
\begin{array}{ll}
\mc{E} & ::= \epsilon ~|~ (\sigma X = f)\ \mc{E}\\
f,g    & ::= c ~|~ X ~|~ f \vee g ~|~ f \wedge g
\end{array}
\]
where  $\epsilon$ is the empty BES, $\sigma {\in} \{\mu,\nu \}$ is a fixed point
symbol, $X$ is a proposition variable taken from some set $\mc{X}$,
$f$ is a proposition formula and $c$ is either constant $\true$ or $\false$.
\end{defn}
For any equation system $\mc{E}$, the set of \emph{bound proposition variables},
$\bnd{\mc{E}}$, is the set of variables occurring at the left-hand side of
some equation in $\mc{E}$. The set of \emph{occurring proposition variables},
$\occ{\mc{E}}$, is the set of variables occurring at the right-hand side of
some equation in $\mc{E}$.
\begin{align*}
\bnd{\epsilon} & \isdef \emptyset &
\bnd{(\sigma X = f)~\mc{E}} & \isdef \bnd{\mc{E}} \cup \{X\}\\
  \occ{\epsilon} & \isdef \emptyset &
  \occ{(\sigma X = f)~\mc{E}} & \isdef \occ{\mc{E}} \cup \occ{f}
\end{align*}
where $\occ{f}$ is defined inductively as follows:
\begin{align*}
  \occ{c} & \isdef \emptyset  & \occ{X}     & \isdef \{X\}   \\
  \occ{f \vee g} & \isdef \occ{f} \cup \occ{g} &
  \occ{f \wedge g} & \isdef \occ{f} \cup \occ{g}
\end{align*}
We say that an equation system $\mc{E}$ is \emph{closed} whenever
$\occ{\mc{E}} \subseteq \bnd{\mc{E}}$.   Intuitively, a (closed) equation
system uniquely assigns truth values to its bound proposition variables,
provided that every bound variable occurs only at the left-hand side
of a single equation in an equation system. In such a case, we call
the equation system \emph{well-formed}. As usual, we only consider
well-formed equation systems.  Well-formedness enables us to define
an ordering $\tle$ on bound variables of an equation system $\mc{E}$,
denoted $X \tle X'$, indicating that the equation for $X$ precedes the
equation for $X'$ in the equation system $\mc{E}$.\\

Formally, proposition formulae are interpreted in a context of an
\emph{environment} $\eta \ap \mc{X} \to \bool$. For an
arbitrary environment $\eta$, we write $\eta [X:=b]$ for the environment
$\eta$ in which the proposition variable $X$ has
Boolean value $b$ and all other proposition variables $X'$ have value $\eta(X')$.
Note that, for reading ease, we do not formally distinguish
between a semantic Boolean value and its representation by $\true$ and $\false$;
likewise, for the operands $\wedge$ and $\vee$.
\begin{defn}\label{def:solution_es}
Let $\eta \ap \mc{X} \to \bool$ be an environment.  The
\emph{interpretation} $\sem{f}{\eta}$ maps a proposition formula $f$
to $\true$ or $\false$:
\begin{align*}
\sem{c}{\eta} &\isdef c & \sem{X}{\eta}      &\isdef \eta(X) \\
\sem{f \vee g}{\eta} &\isdef \sem{f}{\eta} \vee \sem{g}{\eta} &
\sem{f \wedge g}{\eta} &\isdef \sem{f}{\eta} \wedge \sem{g}{\eta}
\end{align*}
The \emph{solution of a BES}, given an environment $\eta$,
is inductively defined as follows:
\[ \begin{array}{lcl}
\sem{\epsilon}{\eta}  & \isdef & \eta \\

\sem{( \sigma X = f )\ \mc{E}}{\eta} & \isdef &
\left \{
\begin{array}{ll}
  \sem{\mc{E}}{ (\eta [X :=  \sem{f}{(\sem{\mc{E}}{\eta[X := \false]})}])}
&  \text{ if $\sigma = \mu$} \\
  \sem{\mc{E}}{ (\eta [X :=  \sem{f}{(\sem{\mc{E}}{\eta[X := \true]})}])}
&  \text{ if $\sigma = \nu$} \\
\end{array}
\right .
\end{array}
\]
\end{defn}
The tree-like recursive definition of a solution makes it intricately
complex. On the one hand, it can be shown that a solution to an equation
system still verifies every equation (in the sense that the value at the
left-hand side is logically equivalent to the value at the right-hand side
of the equation).  At the same time, the fixed-point signs of left-most
equations \emph{outweigh} the fixed-point signs of those equations
that follow, \ie, the fixed-point signs of leftmost equations are more
important. As a consequence, the solution is order-sensitive:
the solution to $(\mu X = Y)\ (\nu Y = X)$, yielding all $\false$,
differs from the solution to $(\nu Y = X)\ (\mu X = Y)$, yielding all
$\true$.

Closed equation systems enjoy the property that the solution to the
equation system is independent of the environment in which it is defined,
\ie, for all environments $\eta,\eta'$, we have $\sem{\mc{E}}{\eta}(X) =
\sem{\mc{E}}{\eta'}(X)$ for all $X \in \bnd{\mc{E}}$.  For this reason,
we henceforth refrain from writing the environment explicitly in all
our considerations dealing with closed equation systems, \ie, we write
$\sem{\mc{E}}$, and $\sem{\mc{E}}(X)$ instead of the more verbose
$\sem{\mc{E}}{\eta}$ and $\sem{\mc{E}}{\eta}(X)$.

An academic example illustrating the typical purpose of equation systems
is given below.
\begin{exmp}
Consider the labeled transition system (depicted below), modelling
mutual exclusion between two readers and a single writer.

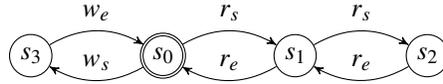
\begin{center}
\begin{tikzpicture}[->,>=stealth',node distance=50pt]
\tikzstyle{every state}=[minimum size=15pt, inner sep=2pt, shape=circle]

\node [state,accepting] (naught) {\small $s_0$};
\node [state] (one) [right of=naught] {\small $s_1$};
\node [state] (two) [right of=one] {\small $s_2$};
\node [state] (three) [left of=naught] {\small $s_3$};

\draw
  (naught) edge[bend left] node [above] {\small $r_s$} (one)
  (one) edge[bend left] node [above]    {\small $r_s$} (two)
  (two) edge[bend left] node [above]    {\small $r_e$} (one)
  (one) edge[bend left] node [above]    {\small $r_e$} (naught)
  (naught) edge[bend left] node [above] {\small $w_s$} (three)
  (three) edge[bend left] node [above]  {\small $w_e$} (naught);
\end{tikzpicture}
\end{center}

\noindent
Reading is started using an action $r_s$
and action $r_e$ indicates its termination. Likewise for writing.
The verification problem $\nu X. \mu
Y. ~\langle r_s\rangle X \vee \langle \overline{r_s}\rangle Y$,
modelling that on some path, a reader can infinitely often start reading,
translates to the following equation
system:

$$
\begin{array}{l}
(\nu X_{s_0} = Y_{s_0})\
(\nu X_{s_1} = Y_{s_1})\
(\nu X_{s_2} = Y_{s_2})\
(\nu X_{s_3} = Y_{s_3})\\
(\mu Y_{s_0} = X_{s_1} \vee Y_{s_1})\
(\mu Y_{s_1} = X_{s_2} \vee Y_{s_0})\
(\mu Y_{s_2} = Y_{s_1})\
(\mu Y_{s_3} = Y_{s_0})
\end{array}
$$
\noindent
Observe that, like the original $\mu$-calculus formula, the resulting equation
system has mutual dependencies between $X$ and $Y$ proposition variables.
Solving the resulting equation system leads to $\true$ for all bound
variables; $X_{s_i} = \true$, for arbitrary state $s_i$,
implies that the property holds in state $s_i$.
\qed
\end{exmp}

The lemma below states that an equation $(\sigma X = f)$ in an equation system
can be moved arbitrarily close to the end in that equation system, so long as all the
proposition variables that
occur in $f$ are bound by equations that precede the equation for $X$. Moreover,
in the special case that $X \notin \occ{f}$, the fixed-point sign of the equation
for $X$ is immaterial, and can thus be changed at will.
\begin{lem}
\label{lem:moving_equation}
Let $\sigma'$ denote
an arbitrary fixed-point sign. If $\occ{f} \cap
\bnd{(\sigma X = f)\ \mc{E}_1\ \mc{E}_2} = \emptyset$, then for all environments $\eta$:
\[
\sem{\mc{E}_0\ (\sigma X = f)\ \mc{E}_1\ \mc{E}_2}{\eta}
=
\sem{\mc{E}_0\ \mc{E}_1\ (\sigma' X = f)\ \mc{E}_2}{\eta}
\]
\end{lem}
\begin{proof} Due to Lemma~3.14 of~\cite{Mad:97}, it suffices to prove the
above equivalence for $\mc{E}_0 = \epsilon$. The resulting equivalence
then follows via an induction on the length of $\mc{E}_1$. The
inductive step is non-trivial. \end{proof}

Note that a variation of the above lemma in which $X \in \occ{f}$ does
not admit a change of fixed-point sign, but, otherwise, the equivalence
still holds.
In several practical and theoretical cases, it suffices to consider
equation systems in which the right-hand sides of the equations are
of a particular shape. The following definition formally introduces
equation systems in \emph{standard recursive form}, which is used
in~\cite{KeirenWillemse2009a}. The introduced syntax takes advantage of
the fact that the semantics of proposition formulae satisfies the usual
rules of Boolean logic such as associativity and commutativity of $\wedge$
and $\vee$.
\begin{defn} A Boolean equation system
$\mc{E}$ in \emph{standard recursive form (SRF)} is defined by the following
grammar:
\[
\begin{array}{ll}
\mc{E} & ::= \epsilon ~|~ (\sigma X = f)\ \mc{E} \\
f      & ::= X ~|~ \bigvee F ~|~ \bigwedge F,
\text{~~~~where $F \subseteq \mc{X}$, with $|F| > 0$.}
\end{array}
\]
The solution to $\mc{E}$ is given by Definition~\ref{def:solution_es},
where proposition formulae in SRF are interpreted as follows (note that
we write $\eta(F)$ to indicate that $\eta$ is applied to every
variable $X$ in $F$):
\begin{align*}
\sem{X}{\eta}  &\isdef \eta(X) &
\sem{\bigwedge F}{\eta} &\isdef \bigwedge \eta(F) &
\sem{\bigvee F}{\eta} &\isdef \bigvee \eta(F)
\end{align*}

\end{defn}
Observe that every equation system $\mc{E}$ can be rewritten to an
equation system $\tilde{\mc{E}}$ in SRF such that $\sem{\mc{E}}{\eta}(X)
= \sem{\tilde{\mc{E}}}{\eta}(X)$ for all $X \in \bnd{\mc{E}}$, \ie,
the transformation to SRF preserves and reflects the solution of bound
variables.  This transformation leads to a polynomial blow-up of the
original equation system.  Lemma~\ref{lem:moving_equation} provides the
foundations for our results in Section~\ref{Sect:Preservation}, where
it underpins the soundness of \emph{normalisation}, \ie, the process of
turning an equation system into SRF.

Next, we consider the \emph{rank} of an
equation system (both standard and in SRF), and the derived notion of
the \emph{alternation hierarchy} of an equation system. The hierarchy can be thought of as
the number of syntactic alternations of fixed point signs occurring
in the equation system. Note that the alternation hierarchy is an
over-approximation of the \emph{alternation depth}, which is a measure
for the complexity of an equation system, measuring the degree of mutual
alternating dependencies. Theoretically, the alternation depth is in
many cases
smaller than the alternation hierarchy; practically, it is harder to
define and compute than the alternation hierarchy.
\begin{defn} Let $\mc{E}$ be an arbitrary equation system. The
\emph{rank} of some $X \in \bnd{\mc{E}}$, denoted
$\rank{X}$, is defined as $\rank{X} = \cad{\nu}{X}{\mc{E}}$, where
$\cad{\nu}{X}{\mc{E}}$ is defined inductively as follows:
\begin{align*}
\cad{\sigma}{X}{\epsilon} & = 0 \\
\cad{\sigma}{X}{(\sigma' Y = f) \mc{E}} & =
\left \{
  \begin{array}{ll}
  0 & \text{ if $\sigma = \sigma'$ and $X = Y$} \\
  \cad{\sigma}{X}{\mc{E}} & \text{ if $\sigma = \sigma'$ and $X \not= Y$} \\
  1+ \cad{\sigma'}{X}{(\sigma' Y = f) \mc{E}} & \text{ if $\sigma \not= \sigma'$} \\
  \end{array}
\right .
\end{align*}
The \emph{alternation hierarchy} $\ah{\mc{E}}$ is the difference between the
maximum and the minimum of the ranks of the equations of $\mc{E}$.
Observe that $\rank{X}$ is odd iff $X$ is defined in a least fixed-point equation.
\end{defn}
The following lemma states that equations with equal ranks can be
switched without affecting the solution. This result is well-known,
and follows from Beki\v c\ principle.
\begin{lem}
\label{lem:switching}
Let $\mc{E}_0\ (\sigma X = f)\ \mc{E}_1\ (\sigma' Y = g)\ \mc{E}_2$
be an arbitrary equation system with $\rank{X} = \rank{Y}$. Then for
arbitrary environment $\eta$, we have:
\[
\sem{\mc{E}_0\ (\sigma X = f)\ \mc{E}_1\ (\sigma' Y = g)\ \mc{E}_2}{\eta}
=
\sem{\mc{E}_0\ (\sigma' Y = g)\ \mc{E}_1\ (\sigma X = f)\ \mc{E}_2}{\eta}
\]
\end{lem}
Finally, for the purpose of comparison with the structure graphs we define
in the next section, we introduce the \emph{dependency graph}
$\langle V, \to,r, l \rangle$ as a derived
notion of an equation system $\mc{E}$
in SRF (see~\cite{KeirenWillemse2009a}), where:
\begin{itemize}
\item $V = \bnd{\mc{E}}$ is a set of nodes;
\item $\to \subseteq V \times V$ is the transition relation, defined as
$X \to Y$ iff $Y \in \occ{f}$ for $\sigma X = f \in \mc{E}$;
\item $r \ap V \to \Nat$ is the rank function,
defined as $r(X) = \rank{X}$;
\item $l \ap V \to  \{\wedge, \vee, \bot \}$ is the logic function, where
$l(X)$ is the Boolean operator in $\sigma X = f \in \mc{E}$, or
$\bot$ if there is no Boolean operator.

\end{itemize}

\section{Structure Graphs for Boolean Equation Systems}
\label{Sect:Structure_Graphs}

A large part of the complexity of equation systems is attributed to the
mutual dependencies between the equations. For closed equation systems in
SRF, these intricate dependencies are captured neatly by the dependency
graphs. For arbitrary equation systems, the situation is more complicated.
We first generalise the notion of a dependency graph to a \emph{structure
graph}, and show that the resulting structure is still adequate for
closed equation systems in SRF. We then proceed to show that arbitrary
non-empty closed equation systems can be mapped onto a structure graph.

\subsection{Structure Graphs}
\label{subsec:sg}

\begin{defn}
A structure graph is a finite, vertex-labeled graph
$\mc{G} = \langle T, t, \to, d \rangle$, where:

\begin{itemize}
\item $T$ is a finite set of proposition formulae;
\item $t \in T$ is the initial formula;
\item $\to \subseteq T \times T$ is a dependency relation;
\item $d\ap T \to (2^{D_\up} \cup 2^{D_\down} \cup 2^{D_{\top}} \cup 2^{D_{\perp}})$, where, for $e \in \{ \up, \down, \top, \perp \}$,
$D_e = \Nat \cup \{e \}$,
is a term decoration mapping;
\end{itemize}
\end{defn}
A structure graph allows for capturing the dependencies between bound
variables and (sub)formulae occurring in the equations of such bound
variables.  Intuitively, the decoration function $d$ reflects the
important information in an arbitrary equation or formula, such as the
ranks of the bound variables, and the indication that the top symbol
of a proposition formula is $\true$ (represented by $\top$), $\false$
(represented by $\perp$), a conjunction (represented by $\up$) or a
disjunction (represented by $\down$). We say some node $t$ is decorated
by some symbol $\star$ whenever $\star \in d(t)$. Our rather liberal choice for the decoration function of nodes  is motivated by possible future extensions of the theory that deal with open equation systems and complex forms of composition; we believe that sets of natural numbers are essential ingredients for accommodating such extensions. Observe that for closed equation systems, at most a single natural number would suffice. One can easily define bisimilarity
on structure graphs.

\begin{defn}
Let $\mc{G} = \langle T, t, \to, d \rangle$ and $\mc{G}' = \langle T', t', \to', d' \rangle$ be structure graphs. A relation $R \subseteq T \times T'$ is a bisimulation relation if for all $(u,u') \in R$
\begin{itemize}
\item $d(u) = d'(u')$;
\item for all $v \in T$, if $u \rightarrow v$, then $u' \rightarrow' v'$ for some $v' \in T'$ such that $(v,v') \in R$;
    \item for all $v' \in T'$, if $u' \rightarrow' v'$, then $u \rightarrow v$ for some $v \in T$ such that $(v,v') \in R$.
\end{itemize}
The structure graphs $\mc{G}$ and $\mc{G}'$ are bisimilar, notation $\mc{G} \bisim \mc{G}'$ if there exists a bisimulation relation $R$ such that $(t,t')\in R$.
\end{defn}

Next, we show how, under some mild conditions, a formula and equation system can be associated to a structure graph. Later in the paper this transformation will be used.

A structure graph $\mc{G} = \langle T,t,\to,d \rangle$ is called \emph{BESsy}~if it satisfies the following five constraints:
\begin{itemize}
\item a node $t$ decorated by $\top$ or $\perp$  has no successor w.r.t.\ $\to$.
\item a node is decorated by $\up$ or $\down$ or a rank iff it has a successor w.r.t.\ $\to$.
\item a node with multiple successors w.r.t.\ $\to$, is decorated with $\up$ or $\down$.
\item a node with rank 0 or 1 is reachable, and the ranks of all reachable nodes form a closed interval.
\item every cycle contains a node with a rank.
\end{itemize}
Observe that BESsyness is preserved under bisimilarity.  For a BESsy
structure graph $\mc{G} = \langle T,t,\to,d \rangle$ the function
$\mathit{term}$ and the partial function $\mathit{rhs}$ are defined
as follows:
\[
\begin{array}{lcl}
\mathit{term}(u) &=&
\begin{cases}
\bigsqcap \{ \mathit{term}(u') \mid u \rightarrow u' \} &
\mbox{if $d(u) = \{ \up \}$}, \\
\bigsqcup \{ \mathit{term}(u') \mid u \rightarrow u' \} &
\mbox{if $d(u) = \{ \down \}$}, \\
\true & \mbox{if $\top \in d(u)$}, \\
\false & \mbox{if $\perp \in d(u)$}, \\
X_{u} & \mbox{otherwise}, \\
\end{cases}
\\
\mathit{rhs}(u) &=&
\begin{cases}
\bigsqcap \{ \mathit{term}(u') \mid u \rightarrow u' \} &
\mbox{if $\up \in d(u)$}, \\
\bigsqcup \{ \mathit{term}(u') \mid u \rightarrow u' \} &
\mbox{if $\down \in d(u)$}, \\
\mathit{term}(u') & \mbox{otherwise, where $u'$ is such that $u \rightarrow u'$}.
\end{cases}
\end{array}
\]
In the definition of the functions $\mathit{term}$ and $\mathit{rhs}$,
the symbols $\bigsqcap$ and $\bigsqcup$ are used as a shorthand for
a nested application of $\land$ and $\lor$. Let $\lessdot$ be a total
order on $\mc{X} \cup \{\true,\false\}$.  Assuming that $\lessdot$
is lifted to a total ordering on formulae, we define for formula $t$
smaller than all formulae in $T$ w.r.t.\ $\lessdot$
\[ \bigsqcap \{ t \} = t \qquad \bigsqcap (\{t\}\cup T) = t \land \left(\bigsqcap T\right)
 \qquad
   \bigsqcup \{ t \} = t \qquad \bigsqcup (\{t\}\cup T) = t \lor \left(\bigsqcup T\right)
\]

\begin{defn}
\label{transSGtoBES}
Let $\mc{G} = \langle T,t,\to,d \rangle$ be a BESsy structure graph.
The formula (and equation system) associated to $\mc{G}$, denoted $\mc{E}_{\mc{G}}$ is the formula $\mathit{term}(t)$ in the context of the equation system $\mc{E}$ defined below.
To each node $u \in T$ such that $d(u) \cap \Nat \neq \emptyset$, we associate an equation of the form \[ \sigma X_u = rhs(u) \] where
$\sigma$ is $\mu$ in case the maximal rank, provided it exists, associated to the node is odd, and $\nu$ otherwise. The equation system $\mc{E}$ is obtained by ordering the equations from left-to-right based on the ranks of the variables.
\end{defn}

\subsection{Structure graphs for equation systems in SRF}

Next, for every formula (not only the variables) in the context of an equation system $\mc{E}$ in SRF, we define the dependency relation and the decorations of formulae denoted by the transition relation $\_\rightarrow\_$ and the predicates $\_\top$, $\_\perp$, $\_\up$, $\_\down$, and $\_\pitchfork n$. It should be noted that for decorating the proposition variables with the rank we use the function $\mathsf{rank}$ that is defined before.
By means of the following deduction rules a structure graph is associated to each formula given a non-empty equation system $\mc{E}$ in SRF:
\[
\sosrule{\sigma X = \bigwedge F \in \mc{E}}{X \up}
\qquad
\sosrule{\sigma X = \bigvee F \in \mc{E}}{X \down}
\qquad
\sosrule{\sigma X = f \in \mc{E} \quad Y \in \occ{f}}{X \rightarrow Y}
\qquad
\sosrule{X\in\bnd{\mc{E}} \quad \rank{X} = n}{X \pitchfork n}\]

\[
\sosrule{}{\bigwedge F \up}
\qquad
\sosrule{}{\bigvee F \down}
\qquad
\sosrule{X \in F}{\bigwedge F \rightarrow X}
\qquad
\sosrule{X \in F}{\bigvee F \rightarrow X}
\]

The structure graph associated to a formula $t$ in the context of an equation system $\mc{E}$ is denoted $\mc{G}_{\mc{E},t}$.
For an equation system $\mc{E}$, let $X \in \bnd{\mc{E}}$ be the least element w.r.t.\ $\tle$. Then, the structure graph associated to $\mc{E}$, denoted by $\mc{G}_{\mc{E}}$, is the structure graph of the variable $X$ in the context of $\mc{E}$.
Structure graphs obtained from the SOS for Boolean equation systems in SRF satisfy the following restrictions.

\begin{lem}\label{lem:SGBESSRF}
Let $\mc{E}$ be a non-empty closed Boolean equation system in SRF and let $\mc{G}_{\mc{E}} = \langle T, X, \to, d \rangle$ be the structure graph associated to $\mc{E}$.
\begin{enumerate}
\item All nodes correspond to propositional variables: For all $t \in T$, we have
$t \in \bnd{\mc{E}}\cup\occ{\mc{E}}$;
\item A node is ranked iff it is a bound variable: For all $t \in T$, we have $d(t) \cap \Nat \neq\emptyset$ iff $t \in \bnd{\mc{E}}$;
\item At most one rank is assigned to a node: For all $t \in T$, we have $|d(t) \cap \Nat| \le 1$.
\end{enumerate}
\end{lem}
\begin{proof}
These properties follow easily from the deduction rules.
\end{proof}

\begin{lem}
For a non-empty closed equation system $\mc{E}$ in SRF the structure graph $\mc{G}_{\mc{E}}$ is isomorphic to the dependency graph defined for it in \cite{KeirenWillemse2009a}.
\end{lem}
\begin{proof}
The properties given in Lemma~\ref{lem:SGBESSRF} precisely characterise the dependency graphs from \cite{KeirenWillemse2009a}.
\end{proof}

\subsection{Structure graphs for non-empty closed equation systems}

Next, we define structure graphs not for the class of equation systems in SRF but for arbitrary closed non-empty equation systems. First, as before, nodes representing bound propositional variables are labeled by a natural number representing the rank of the variable in the equation system:

\[ \sosrule{X\in\bnd{\mc{E}} \quad \rank{X} = n}{X \pitchfork n}\ . \]

A clear difference between equation systems in SRF and the more general
class of equation systems is that in the latter only a binary version
of conjunction and disjunction is available. A question that needs
to be answered is ``How to capture this structure in the structure
graph?'' One way of doing so would be to precisely reflect the structure
of the right-hand side. For a right-hand side of the form $X \land (Y
\land Z)$ this results in the structure graph depicted below (left),
where we assume that the ranks of the variables $X$, $Y$, and $Z$
are 1, 2, and 3, respectively:

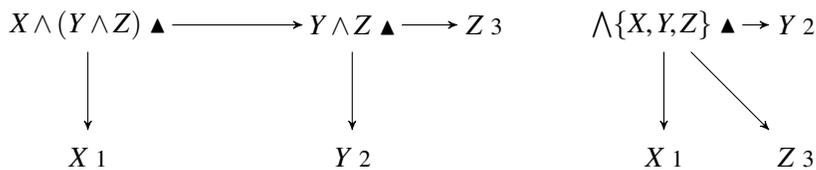
\begin{tabular}{c@{\hspace*{1cm}}c}
\begin{tikzpicture}[->,>=stealth',node distance=50pt]
\tikzstyle{every state}=[draw=none,minimum size=20pt, inner sep=2pt, shape=rectangle]

\node[state] (f) {$X \land (Y \land Z)\ \text{\footnotesize{$\up$}}$};
\node[state] (fake) [right of=f] {};
\node[state] (X) [below of=f] {$X\ \text{\footnotesize{$1$}}$};
\node[state] (sf) [right of=fake] {$Y \land Z\ \text{\footnotesize{$\up$}}$};
\node[state] (Y) [below of=sf] {$Y\ \text{\footnotesize{$2$}}$};
\node[state] (Z) [right of=sf] {$Z\ \text{\footnotesize{$3$}}$};

\draw (f) edge (X)
      (f) edge (sf)
      (sf) edge (Y)
      (sf) edge (Z);
\end{tikzpicture}
&
\begin{tikzpicture}[->,>=stealth',node distance=50pt]
\tikzstyle{every state}=[draw=none,minimum size=20pt, inner sep=2pt, shape=rectangle]

\node[state] (f) {$\bigwedge \{X,Y,Z\}\ \text{\footnotesize{$\up$}}$};
\node[state] (X) [below of=f] {$X\ \text{\footnotesize{$1$}}$};
\node[state] (Y) [right of=f] {$Y\ \text{\footnotesize{$2$}}$};
\node[state] (Z) [below of=Y] {$Z\ \text{\footnotesize{$3$}}$};

\draw (f) edge (X) edge (Y) edge (Z);
\end{tikzpicture}
\end{tabular}

A drawback of this solution is that, in general, the logical equivalence
between $\bigwedge \{ X,Y,Z \}$ and the formula $X \land (Y \land Z)$
is not reflected by bisimilarity. Retaining this logical equivalence
(and hence associativity and commutativity) of both conjunction and
disjunction is desirable to approximate the power of dependency graphs
in reducing w.r.t.\ bisimilarity.

Another syntactic difference between equation systems in SRF described in
\cite{KeirenWillemse2009a} and the more general class of equation systems
discussed in this paper is that the logical connectives for conjunction
($\land$) and disjunction ($\lor$) may occur nested in the right-hand
side of the same Boolean equation. This is solved by reflecting a change
in leading operator in the structure graph. So the anticipated structure
of the structure graph for $X \land (Y \land (Z \lor X))$, where, again,
we assume that the ranks of the variables $X$, $Y$, and $Z$ are 1, 2,
and 3, respectively, is:

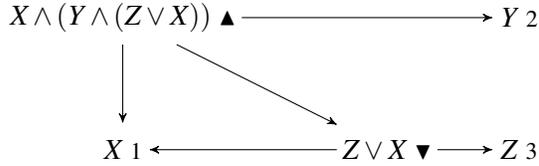
\begin{tikzpicture}[->,>=stealth',node distance=50pt]
\tikzstyle{every state}=[draw=none,minimum size=20pt, inner sep=2pt, shape=rectangle]

\node[state] (f) {$X \land (Y \land (Z \lor X))\ \text{\footnotesize{$\up$}}$};
\node[state] (X) [below of=f] {$X\ \text{\footnotesize{$1$}}$};
\node[state] (sf) [below of=sf] {$Z \lor X\ \text{\footnotesize{$\down$}}$};
\node[state] (Z) [right of=sf] {$Z\ \text{\footnotesize{$3$}}$};
\node[state] (Y) [above of=Z] {$Y\ \text{\footnotesize{$2$}}$};

\draw (f) edge (X)
      (f) edge (Y)
      (f) edge (sf)
      (sf) edge (Z) edge (X);
\end{tikzpicture}

This can be elegantly achieved by means of the following deduction rules for the decorations and the dependency transition relation $\rightarrow$:
\[
\sosrule{}{\true \top}
\qquad
\sosrule{}{\false \perp}
\qquad
\sosrule{}{t \land t'\up}
\qquad
\sosrule{}{t \lor t'\down}
\qquad
\sosrule{t\up \quad t \rightarrow u}{t \land t' \rightarrow u}
\qquad
\sosrule{t'\up \quad t' \rightarrow u'}{t \land t' \rightarrow u'} \]
\[
\sosrule{t\down \quad t \rightarrow u}{t \lor t' \rightarrow u}
\qquad
\sosrule{t'\down \quad t' \rightarrow u'}{t \lor t' \rightarrow u'}
\qquad
\sosrule{\neg t \up}{t \land t' \rightarrow t}
\qquad
\sosrule{\neg t' \up}{t \land t' \rightarrow t'}
\qquad
\sosrule{\neg t \down}{t \lor t' \rightarrow t}
\qquad
\sosrule{\neg t' \down}{t \lor t' \rightarrow t'}
\]
The first four deduction rules for $\rightarrow$ are introduced to flatten the nesting hierarchy of the same connective. They can be used to deduce that $X \land (Y \land Z) \rightarrow Y$. The latter four deduction rules describe the dependencies in case there is no flattening possible anymore (by absence of structure). For example $X \land Y \rightarrow X$ is derived by means of the first of these deduction rules.

\begin{exmp}
The proposition formula $(X \land (Z \lor (Y \lor X))) \land Z$ results in the following structure graph fragment. The subgraphs generated by the equations for $X$, $Y$, and $Z$ are omitted from this example.

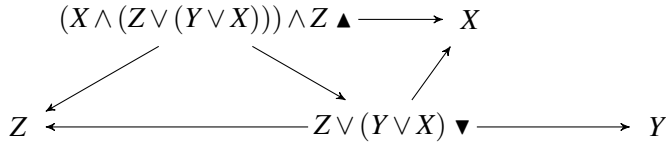
\begin{tikzpicture}[->,>=stealth',node distance=100pt]
\tikzstyle{every state}=[draw=none,minimum size=20pt, inner sep=2pt, shape=rectangle]

\node[state] (f) {$(X \land (Z \lor (Y \lor X))) \land Z\ \text{\footnotesize{$\up$}}$};
\node[state] (X)  [right of=f] {$X$};
\node[state] (ssf) [below right of=f, yshift=30pt] {$Z \lor (Y \lor X)\ \text{\footnotesize{$\down$}}$};
\node[state] (Z) [below left of=f, yshift=30pt] {$Z$};
\node[state] (Y) [right of=ssf] {$Y$};

\draw (f) edge (X) edge (Z) edge (ssf)
      (ssf)  edge (Z) edge (Y) edge (X);
\end{tikzpicture}
\end{exmp}

It should be noted that all these predicates and transitions are defined in the context of one and the same equation system.

Finally, we present deduction rules that describe how the structure of a node representing a variable is derived from the right-hand side of the corresponding equation. The third deduction rule defines this for the case that the right-hand side is a variable, the last two deduction rule for the cases it is a proposition formula that is not a variable.

\[
\sosrule{\sigma X = t \in \mc{E} \quad t \down}{X \down}
\qquad
\sosrule{\sigma X = t \in \mc{E} \quad t \up}{X \up}
\qquad
\sosrule{\sigma X = t \in \mc{E} \quad \neg t \up \quad \neg t \down}{X \rightarrow t}
\]
\[
\sosrule{\sigma X = t \in \mc{E} \quad t \down \quad t \rightarrow u}{X \rightarrow u}
\qquad
\sosrule{\sigma X = t \in \mc{E} \quad t \up \quad t \rightarrow u}{X \rightarrow u}
\]

\begin{exmp}\label{ex:structure_graph}  An equation system (see left)
and its associated structure graph (see right). Observe that the term $X
\wedge Y$ is shared by the equations for $X$ and $Y$, and appears only
once as a node in the structure graph as an unranked node. The equation
for $Z$ is represented by term $Z$, and is decorated only by the rank of
the equation for $Z$.  The subterm $Z \vee W$ in the equation for $W$
does not appear as a separate node in the structure graph, since the
disjunctive subterm occurs within the scope of another disjunction.

\noindent
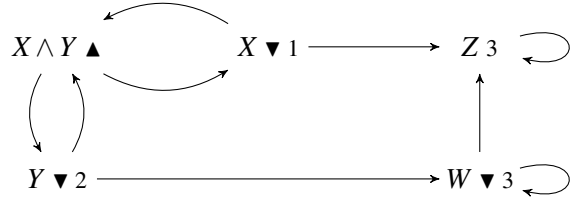
\parbox{.35\textwidth}
{
\[
\begin{array}{lcl}
\mu X &=& (X \wedge Y) \vee Z\\
\nu Y &=& W \vee (X \wedge Y)\\
\mu Z &=& Z\\
\mu W &=& Z \vee (Z \vee W)
\end{array}
\]
}
\qquad
\parbox{.55\textwidth}
{
\begin{tikzpicture}[->,>=stealth',node distance=50pt]
\tikzstyle{every state}=[shape=rectangle,draw=none,minimum width=30pt, minimum height=20pt, inner sep=2pt]

\node[state] (X) {$X\ \text{\footnotesize{$\down\ 1$}}$};
\node[state] (Z) [right of=X,xshift=30pt] {$Z\ \text{\footnotesize{$3$}}$};
\node[state] (XY) [left of=X,xshift=-30pt] {$X \wedge Y\ \text{\footnotesize{$\up$}}$};
\node[state] (Y) [below of=XY] {$Y\ \text{\footnotesize{$\down\ 2$}}$};
\node[state] (W) [below of=Z] {$W\ \text{\footnotesize{$\down\ 3$}}$};

\draw (X) edge  (Z) edge[bend right] (XY)
      (Y) edge (W) edge [bend right] (XY)
      (Z) edge [loop right] (Z)
      (XY) edge [bend right] (Y) edge [bend right] (X)
      (W) edge (Z) edge [loop right] (W);
\end{tikzpicture}
}

\qed
\end{exmp}

The structure graph associated to a formula $t$ in the context of an equation system $\mc{E}$ is denoted $\mc{G}_{\mc{E},t}$. For an equation system $\mc{E}$, let $X \in \bnd{\mc{E}}$ be the least element w.r.t.\ $\tle$. Then, the structure graph associated to $\mc{E}$, denoted by $\mc{G}_{\mc{E}}$, is the structure graph of the variable $X$ in the context of $\mc{E}$.

\begin{lem}
Let $\mc{E}$ be a non-empty closed equation system. Let $t$, $t'$, and $t''$ be arbitrary proposition formulae such that $\occ{t} \cup \occ{t'} \cup \occ{t''} \subseteq \bnd{\mc{E}}$. Then the following hold:
\[ \mc{G}_{\mc{E},(t \land t') \land t''} \bisim \mc{G}_{\mc{E},t \land (t' \land t'')},
\quad \mc{G}_{\mc{E},(t \lor t') \lor t''} \bisim \mc{G}_{\mc{E},t \lor (t' \lor t'')},
\quad \mc{G}_{\mc{E},t \land t'} \bisim \mc{G}_{\mc{E},t' \land t},
\quad \mc{G}_{\mc{E},t \lor t'} \bisim \mc{G}_{\mc{E},t' \lor t}
\]
\end{lem}

\begin{proof}
The proofs are easy. For example, the bisimulation relation that witnesses
bisimilarity of $(t \land t') \land t''$ and $t \land (t' \land t'')$
is the relation that relates all formulae of the form $(u \land u')
\land u''$ and $u \land (u' \land u'')$ and additionally contains the
identity relation on formulae. Proofs of the ``transfer conditions''
are easy as well. As an example, suppose that $(u \land u') \land u''
\rightarrow v$ for some formula $v$. In case this transition is due to
$u \land u' \up$ and $u \land u' \rightarrow v$, one of the cases that occurs for $u
\land u' \rightarrow v$ is that $u \up$ and $u \rightarrow v$. We obtain
$u \land (u' \land u'') \rightarrow v$. Since $v$ and $v$ are related,
this finishes the proof of the transfer condition in this case. All
other cases are similar or at least equally easy.
\end{proof}

Idempotency of $\land$ and $\lor$, and more involved logical equivalences
such as distribution and absorption are not captured by isomorphism or
even bisimilarity on the structure graphs. The reason is that a right-hand
side $X \land X$ will be decorated by $\up$, whereas a right-hand side
$X$ is not!

\begin{thm}
\label{thm:com}
Let $\mc{E}$ be a non-empty, closed equation system, and let
$\mc{E'}$ be the equation system obtained by transforming
$\mc{G}_{\mc{E}}$ into an equation system.
Then there is a
total bijective mapping $h : \bnd{\mc{E}} \to \bnd{\mc{E}'}$
such that for all $X \in \bnd{\mc{E}}$:
$\sem{\mc{E}}{}(X) = \sem{\mc{E}'}{}(h(X))$.
\end{thm}

\begin{proof}
The mapping $h$ that maps variable $Y \in \bnd{\mc{E}}$ to the variable
$X_Y \in \bnd{\mc{E'}}$ is such a total bijective mapping. The equalities
$\sem{\mc{E}}{}(X) = \sem{\mc{E}'}{}(h(X))$ (for $X \in \bnd{\mc{E}}$)
follow from the construction described in Section \ref{subsec:sg}.
\end{proof}

Next, we study the relationship between the structure graphs as used
for the Boolean equation systems in SRF and Boolean equation systems.
Given a total order $\lessdot$ on $\mc{X}$, there is an embedding $\hbar$
of formulae in the syntax of the right-hand sides of equation systems
in SRF in formulae in the syntax of equation systems.
\[
\begin{array}{lcl}
\hbar(X) = X; \\
\hbar(\bigwedge \{ X\}) = X \land X; & \qquad & \hbar(\bigvee \{ X\}) = X \lor X; \\
\hbar(\bigwedge F) = \min(F) \land \hbar(\bigwedge F\setminus\{\min(F)\});
&& \hbar(\bigvee F) = \min(F) \lor \hbar(\bigvee F\setminus\{\min(F)\}); \\
\end{array}
\]
where $F \subseteq \mc{X}$ such that $|F| \geq 2$  and $\min(F)$ denotes
the least element of $F$ w.r.t\ $\lessdot$. This embedding is easily
lifted to the level of the equation systems themselves.  Note that
an artefact of the above transformation is that the right-hand side
$\bigwedge \{X,Y\}$, assuming that $X \lessdot Y$, is transformed into
$X \land (Y \land Y)$.

\begin{thm}
For an equation system $\mc{E}$ in SRF, $\mc{G}_{\mc{E}} \bisim \mc{G}_{\hbar(\mc{E})}$.
\end{thm}

\begin{proof}
The mapping $\hbar$ (as a relation) is a bisimulation relation that proves $\mc{G}_{\mc{E}} \bisim \mc{G}_{\hbar(\mc{E})}$.
\end{proof}

\section{Preservation and Reflection of Solution under Bisimilarity}
\label{Sect:Preservation}

In the previous section, we showed that there is a natural correspondence
between structure graphs for equation systems in SRF and their
dependency graphs.  We tighten this result by showing that bisimilarity
on structure graphs is a congruence for normalisation (a process similar
to the transformation of an equation system into SRF). As a consequence,
this result allows us to reuse the results of~\cite{KeirenWillemse2009a}
and prove that for each pair of bisimilar nodes,
both nodes have the same truth-value.

\subsection{$\true$/$\false$-Elimination}

Since we strive to reuse some of the results pertaining the dependency graphs, it is useful to define an operator on equation systems that replaces occurrences of nodes representing $\true$ or $\false$ by nodes representing proposition variables. This operator $\mathsf{reduce}$ is defined by the following deduction rules.

\[ \sosrule{t\up}{\reduce{t} \up}
\qquad
\sosrule{t\down}{\reduce{t} \down}
\qquad
\sosrule{\neg t \top \quad \neg t \perp \quad t \rightarrow u}{\reduce{t} \rightarrow \reduce{u}}
\qquad
\sosrule{\neg t \top \quad \neg t \perp \quad t \pitchfork n}{\reduce{t} \pitchfork n}
\]
\[
\sosrule{t \top}{\reduce{t} \rightarrow \reduce{t}}
\qquad
\sosrule{t \top}{\reduce{t} \pitchfork 0}
\qquad
\sosrule{t \perp}{\reduce{t} \rightarrow \reduce{t}}
\qquad
\sosrule{t \perp}{\reduce{t} \pitchfork 1}
\]
Observe that bisimilarity is a congruence for $\mathsf{reduce}$, and
that the operation preserves and reflects the solution of the original
equation system.  In the remainder of Section~\ref{Sect:Preservation},
we assume that operation $\mathsf{reduce}$ has been applied to all
equations systems.

\subsection{Normalisation}

In structure graphs underlying an equation system, terms that are
decorated by ranks typically occur as left-hand sides in equations,
whereas the non-ranked terms occur as subterms in right-hand sides of
equations with mixed occurrences of $\wedge$ and $\vee$. Normalisation
of an equation system can be achieved by introducing a new equation
for subterms in which the top-level Boolean operator differs from the
top-level operator of the term it occurs in: an equation $\sigma X =
Y \wedge (Z \vee W)$ in an equation system is turned into $\sigma X
= Y \wedge Z'$, and an additional equation $\sigma' Z' = Z \vee W$
is introduced in the equation system. In choosing the location (or,
formally, its rank) of this new equation, one has some degree of freedom
(see Lemma~\ref{lem:moving_equation}). We choose to assign a rank based on
the rank that is assigned to its successors in the structure graph. This
is formalised by the following set of deduction rules:

\[
\sosrule{t \up}{\normalise{t}\up}
\qquad
\sosrule{t \down}{\normalise{t}\down}
\qquad
\sosrule{t \rightarrow u}{\normalise{t} \rightarrow \normalise{u}}
\qquad
\sosrule{t \pitchfork n}{\normalise{t} \pitchfork n}
\]
\[
\sosrule{t \not\pitchfork \quad t \rightarrow u \quad \normalise{u} \pitchfork n \quad \forall_{v}~ t \rightarrow v \Rightarrow \normalise{v} \pitchfork m \land m \leq n}
        {\normalise{t} \pitchfork n}
 \]
The last deduction rule expresses that in case a node $t$ does not have
a rank, a rank is associated to the normalised version of $t$. This
rank is the maximal rank of all successors of $t$ (after ranking these
as well). Observe that, by construction, a non-ranked node can not have a transition to itself.
Note that the premise of the last deduction rule is not only
expressed in terms of transitions, predicates and negative versions
thereof, but also utilises logical connectives and even a universal
quantification. Syntax and semantics of such complex premises are taken
from~\cite{WeerdenburgReniers2008}.

Normalisation typically preserves and reflects the solution to an equation
system, in the sense that the Boolean value of all proposition variables,
bound in the original equation system, remains unchanged by the
operation (modulo naming of the proposition variables). This is formalised
by the lemma below. $\mc{G}_{\normalise{\mc{E}}}$ denotes the structure graph $\mc{G}_{\mc{E},\normalise{X}}$ where $X$ is the least variable w.r.t.\ $\tle$.
\begin{lem}
\label{lem:normalisation}
Let $\mc{E}$ be a non-empty, closed equation system, and let
$\mc{E}_\textsf{norm}$ be the equation system obtained by transforming
$\mc{G}_{\normalise{\mc{E}}}$ into an equation system. Then there is a
total injective mapping $h : \bnd{\mc{E}} \to \bnd{\mc{E}_\textsf{norm}}$
such that for all $X \in \bnd{\mc{E}}$:
$\sem{\mc{E}}{}(X) = \sem{\mc{E}_\textsf{norm}}{}(h(X))$.
\end{lem}
\begin{proof} The proof follows from the deduction rules, in combination with
Lemmata~\ref{lem:moving_equation} and~\ref{lem:switching}.\end{proof}

\begin{exmp} The structure graph of the equation system of
Example~\ref{ex:structure_graph} (see left below), and the structure graph of
the normalisation of the same equation system (see right below). Observe that
the term $Y \wedge X$, which, before normalisation is unranked, inherits
the maximal rank of successors $Y$ and $X$.

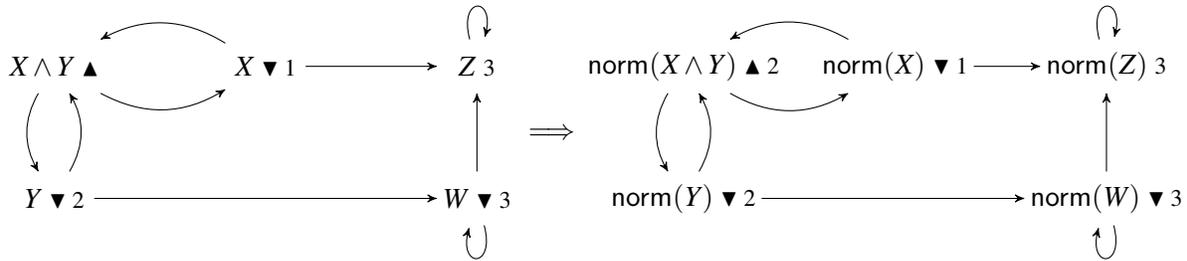
\begin{center}
\parbox{.43\textwidth}
{
\begin{tikzpicture}[->,>=stealth',node distance=50pt]
\tikzstyle{every state}=[shape=rectangle,draw=none,minimum width=30pt, minimum height=20pt, inner sep=2pt]

\node[state] (X) {$X\ \text{\footnotesize{$\down\ 1$}}$};
\node[state] (Z) [right of=X,xshift=30pt] {$Z\ \text{\footnotesize{$3$}}$};
\node[state] (XY) [left of=X,xshift=-30pt] {$X \wedge Y\ \text{\footnotesize{$\up$}}$};
\node[state] (Y) [below of=XY] {$Y\ \text{\footnotesize{$\down\ 2$}}$};
\node[state] (W) [below of=Z] {$W\ \text{\footnotesize{$\down\ 3$}}$};

\draw (X) edge  (Z) edge[bend right] (XY)
      (Y) edge (W) edge [bend right] (XY)
      (Z) edge [loop above] (Z)
      (XY) edge [bend right] (Y) edge [bend right] (X)
      (W) edge (Z) edge [loop below] (W);

\end{tikzpicture}
}
$\Longrightarrow$
\parbox{.51\textwidth}
{
\begin{tikzpicture}[->,>=stealth',node distance=50pt]
\tikzstyle{every state}=[shape=rectangle,draw=none,minimum width=30pt, minimum height=20pt, inner sep=2pt]

\node[state] (X) {$\normalise{X}\ \text{\footnotesize{$\down\ 1$}}$};
\node[state] (Z) [right of=X,xshift=30pt] {$\normalise{Z}\ \text{\footnotesize{$3$}}$};
\node[state] (XY) [left of=X,xshift=-30pt] {$\normalise{X \wedge Y}\ \text{\footnotesize{$\up\ 2$}}$};
\node[state] (Y) [below of=XY] {$\normalise{Y}\ \text{\footnotesize{$\down\ 2$}}$};
\node[state] (W) [below of=Z] {$\normalise{W}\ \text{\footnotesize{$\down\ 3$}}$};

\draw (X) edge  (Z) edge[bend right] (XY)
      (Y) edge (W) edge [bend right] (XY)
      (Z) edge [loop above] (Z)
      (XY) edge [bend right] (Y) edge [bend right] (X)
      (W) edge (Z) edge [loop below] (W);

\end{tikzpicture}
}
\end{center}
According to Lemma~\ref{lem:normalisation}, the equations for
nodes $X$ and $\normalise{X}$ have the same solution.
\qed
\end{exmp}

\begin{lem}
\label{lem:norm_preserv_bisimilarity}
Let $\mc{E}$ and $\mc{E}'$ be non-empty closed equation systems. If $\mc{G}_{\mc{E}}$ and $\mc{G}_{\mc{E}'}$ are bisimilar, then also $\mc{G}_{\normalise{\mc{E}}}$ and $\mc{G}_{\normalise{\mc{E}'}}$ are bisimilar.
\end{lem}

\begin{proof}
Any bisimulation relation $\rel$ witnessing
$\mc{G}_{\mc{E}} \bisim \mc{G}_{\mc{E}'}$ induces a witness for
$\mc{G}_{\normalise{\mc{E}}} \bisim \mc{G}_{\normalise{\mc{E}'}}$.
\end{proof}

\begin{lem}
\label{lem:bisimilarity_srf_non-srf}
Let $\mc{E}$ be a non-empty closed equation system. Then there is an
equation system $\mc{E}'$ in SRF with
$\mc{G}_{\normalise{\mc{E}}} \bisim \mc{G}_{\mc{E}'}$.
\end{lem}

\begin{proof}
The structure graph of $\mc{E}$ is easily transformed into an equation
system in SRF as described previously. Observe that, since all nodes of
$\mc{G}_{\normalise{\mc{E}}}$ are ranked, each equation $(\sigma X = f)$
in $\mc{E}'$ has a right-hand side formula $f$ with at most one type of
Boolean operator, and the structure graph of a non-empty, closed equation
system is BESsy by definition.
\end{proof}

\subsection{Bisimilarity Implies Solution Equivalence}

The theorem below states our main result, proving that equations that
induce bisimilar structure graphs essentially have the same solution. This
allows one to safely use bisimulation minimisation of the underlying
structure graph of an equation system, and solve the resulting equation
system instead. The proof of this theorem relies on the connections
between normalisation, equation systems in SRF, and the results
of~\cite{KeirenWillemse2009a}.

\begin{thm}
Let $\mc{E}$ and $\mc{E}'$ be non-empty, closed equation systems. Then
for every pair of bisimilar formulas $f$ w.r.t.\ $\mc{G}_{\mc{E}}$ and $f'$ w.r.t.\ $\mc{G}_{\mc{E}'}$, also $\sem{f}{}\sem{\mc{E}}{} = \sem{f'}{}\sem{\mc{E}'}{}$.
\end{thm}
\begin{proof}
By Lemma~\ref{lem:norm_preserv_bisimilarity}, it follows that
for each pair $f$ w.r.t.\ $\mc{G}_{\mc{E}}$ and $f'$ w.r.t.\ $\mc{G}_{\mc{E}'}$
of bisimilar nodes, the nodes $\normalise{f}$ and $\normalise{f'}$ are
bisimilar. As a consequence of Lemma~\ref{lem:bisimilarity_srf_non-srf},
we find that there must exist a closed equation system in SRF with a
structure graph that is bisimilar to $\normalise{f}$.  Likewise for
$\normalise{f'}$.  Since $\bisim$ is an equivalence relation, the
structure graphs of the equation systems in SRF are again bisimilar. By
Theorem~1 in~\cite{KeirenWillemse2009a}, we find that this implies that
$\normalise{f}$ and $\normalise{f'}$ have the same solution. Since
normalisation preserves and reflects the solution of the original
equation system, see Lemma~\ref{lem:normalisation}, we find that $f$
and $f'$ have the same solution.
\end{proof}

\section{Application}
\label{Sect:Application}

Equation systems with non-trivial right-hand sides (\ie, equation systems
with equations with right-hand sides containing both $\wedge$ and $\vee$)
occur naturally in the context of process equivalence checking problems
such as the branching bisimulation problem (see \eg~\cite{CPPW:07})
and the more involved model checking problems.  As a slightly more
elaborate example of the latter, we consider a $\mu$-calculus model
checking problem involving an unreliable channel. The channel can read
messages from the environment, and send or lose these next.  In case
the message is lost, subsequent attempts are made to send the message
until this finally succeeds. The labeled transition system, modelling
this system is given below.

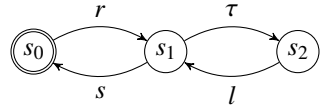
\begin{center}
\begin{tikzpicture}[->,>=stealth',node distance=50pt]
\tikzstyle{every state}=[minimum size=15pt, inner sep=2pt, shape=circle]

\node [state,accepting] (r) {\small $s_0$};
\node [state] (s) [right of=r] {\small $s_1$};
\node [state] (l) [right of=s] {\small $s_2$};

\draw
  (r) edge[bend left] node [above] {\small $r$} (s)
  (s) edge[bend left] node [below]    {\small $s$} (r)
  (s) edge[bend left] node [above]    {\small $\tau$} (l)
  (l) edge[bend left] node [below]    {\small $l$} (s);
\end{tikzpicture}
\end{center}

Suppose we wish to know for which states it holds whether along all
paths consisting of reading and sending actions, it is infinitely often
possible to potentially never perform a send action. Intuitively, this
should be the case in all states: from states $s_0$ and $s_1$, there
is a finite path leading to state $s_1$, which can subsequently produce
the infinite path $(s_1\ s_2)^\omega$, along which the send action does
not occur.  For state $s_2$, we observe that there is no path consisting
of reading and writing actions, so the property holds vacuously in $s_2$.
We formalise this problem as follows:\footnote{Alternative phrasings
are possible, but this one nicely projects onto an equation system with
non-trivial right-hand sides, clearly illustrating the theory outlined
in the previous sections in an example of manageable proportions.}
\[
\phi \equiv
\nu X. \mu Y. ( ([r]X \wedge [s]X \wedge (\nu Z. \langle \overline{s} \rangle Z) ) \vee ([r]Y \wedge [s]Y))
\]
Using the translation of Mader~\cite{Mad:97} of the model checking problem into equation
systems, the equation system given below is obtained. The solution to $X_{s_i}$
answers whether $s_i \models \phi$.
$$
\begin{array}{l}
(\nu X_{s_0} = Y_{s_0})\
(\nu X_{s_1} = Y_{s_1})\
(\nu X_{s_2} = Y_{s_2})\\

(\mu Y_{s_0} = (X_{s_1} \wedge Z_{s_0}) \vee Y_{s_1})\
(\mu Y_{s_1} = (X_{s_0} \wedge Z_{s_1}) \vee Y_{s_0})\
(\mu Y_{s_2} = \true)\\

(\nu Z_{s_0} = Z_{s_1})\
(\nu Z_{s_1} = Z_{s_2})\
(\nu Z_{s_2} = Z_{s_1})
\end{array}
$$
The structure graph underlying the above equation system, restricted to
those parts reachable from the bound variables of the equation system,
is depicted below:
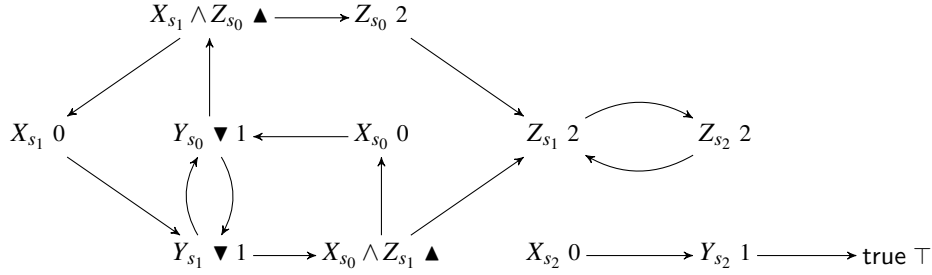
\begin{center}
\begin{tikzpicture}[->,>=stealth',node distance=65pt]
\tikzstyle{every state}=[draw=none,minimum size=15pt, inner sep=2pt, shape=rectangle]

\node [state] (Xs0)                            {\small $X_{s_0}$ \footnotesize $0$};
\node [state] (Ys0)          [left of=Xs0]     {\small $Y_{s_0}$ \footnotesize $\down\ 1$};
\node [state] (Xs1Zs0)       [above of=Ys0,yshift=-20pt]    {\small $X_{s_1} \wedge Z_{s_0}$ \footnotesize $\up$};
\node [state] (Ys1)          [below of=Ys0,yshift=20pt]    {\small $Y_{s_1}$ \footnotesize $\down\ 1$};
\node [state] (Xs1)          [left of=Ys0]     {\small $X_{s_1}$ \footnotesize $0$};
\node [state] (Xs0Zs1)       [right of=Ys1]    {\small $X_{s_0} \wedge Z_{s_1}$ \footnotesize $\up$};
\node [state] (Zs0)          [right of=Xs1Zs0] {\small $Z_{s_0}$ \footnotesize $2$};
\node [state] (Zs1)          [right of=Xs0]    {\small $Z_{s_1}$ \footnotesize $2$};
\node [state] (Zs2)          [right of=Zs1]    {\small $Z_{s_2}$ \footnotesize $2$};
\node [state] (Xs2)          [below of=Zs1,yshift=20pt]    {\small $X_{s_2}$ \footnotesize $0$};
\node [state] (Ys2)          [right of=Xs2]    {\small $Y_{s_2}$ \footnotesize $1$};
\node [state] (Wtrue)        [right of=Ys2]    {\small $\true$ \footnotesize $\top$};

\draw
   (Xs0) edge (Ys0)
   (Ys0) edge (Xs1Zs0) edge[bend left] (Ys1)
   (Ys1) edge[bend left] (Ys0) edge (Xs0Zs1)
   (Xs1Zs0) edge (Xs1)
   (Xs1) edge (Ys1)
   (Xs0Zs1) edge (Xs0) edge (Zs1)
   (Xs1Zs0) edge (Zs0)
   (Zs0) edge (Zs1)
   (Zs1) edge[bend left] (Zs2)
   (Zs2) edge[bend left] (Zs1)
   (Xs2) edge (Ys2)
   (Ys2) edge (Wtrue)
   ;
\end{tikzpicture}
\end{center}
Observe that we have $X_{s_0} \bisim X_{s_1}$, $Z_{s_0} \bisim Z_{s_1}
\bisim Z_{s_2}$, $Y_{s_0} \bisim Y_{s_1}$ and $X_{s_0} \wedge
Z_{s_1} \bisim X_{s_1} \wedge Z_{s_0}$.  Minimising the above structure
graph with respect to bisimulation leads to the structure graph depicted
below:

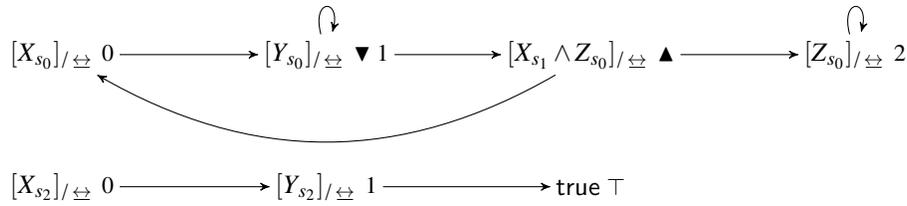
\begin{center}
\begin{tikzpicture}[->,>=stealth',node distance=100pt]
\tikzstyle{every state}=[draw=none,minimum size=15pt, inner sep=2pt, shape=rectangle]

\node [state] (Xs0_r)                            {\small $[X_{s_0}]_{/\bisim}$ \footnotesize $0$};
\node [state] (Ys0_r)      [right of=Xs0_r]  {\small $[Y_{s_0}]_{/\bisim}$ \footnotesize $\down\ 1$};
\node [state] (Xs1Zs0_r)   [right of=Ys0_r]  {\small $[X_{s_1} \wedge Z_{s_0}]_{/\bisim}$ \footnotesize $\up$};
\node [state] (Xs2_r)        [below of=Xs0_r,yshift=50pt]    {\small $[X_{s_2}]_{/\bisim}$ \footnotesize $0$};
\node [state] (Ys2_r)        [right of=Xs2_r]    {\small $[Y_{s_2}]_{/\bisim}$ \footnotesize $1$};
\node [state] (Zs0_r)      [right of=Xs1Zs0_r] {\small $[Z_{s_0}]_{/\bisim}$ \footnotesize $2$};
\node [state] (Wtrue)        [right of=Ys2_r]    {\small $\true$ \footnotesize $\top$};

\draw
   (Xs0_r) edge (Ys0_r)
   (Ys0_r) edge (Xs1Zs0_r) edge[loop above] (Ys0_r)
   (Xs1Zs0_r) edge (Zs0_r) edge[bend left] (Xs0_r)
   (Xs2_r) edge (Ys2_r)
   (Zs0_r) edge [loop above] (Zs0_r)
   (Ys2_r) edge (Wtrue);
\end{tikzpicture}
\end{center}
Note that the structure graph is BESsy, and, hence, admits a translation
back to an equation system. Using the translation provided in
Definition~\ref{transSGtoBES} results in the following equation system:
\[
\begin{array}{l}
(\nu X_{[X_{s_0}]_{/\bisim}} = X_{[Y_{s_0}]_{/\bisim}})\
(\nu X_{[X_{s_2}]_{/\bisim}} = X_{[Y_{s_2}]_{/\bisim}})\\
(\mu X_{[Y_{s_0}]_{/\bisim}} = (X_{[X_{s_0}]_{/\bisim}} \wedge X_{[Z_{s_0}]_{/\bisim}}) \vee X_{[Y_{s_0}]_{/\bisim}})\
(\mu X_{[Y_{s_2}]_{/\bisim}} = \true )\\
(\nu X_{[Z_{s_0}]_{/\bisim}} = X_{[Z_{s_0}]_{/\bisim}})
\end{array}
\]
Answering the global model checking problem can thus be achieved
by solving 6 equations rather than the original 9 equations. Using
standard algorithms for solving equation systems, one quickly finds
that all variables $X_{[V]_{/\bisim}}$ of the minimised equation system
(and thus all nine original proposition variables) have value $\true$.
Note that the respective sizes of the equation systems are 26 before
minimisation and 14 after minimisation, which is slightly less than
a 50\% gain. Such gains appear to be typical in this setting (see
also~\cite{KeirenWillemse2009a}), and seem to surpass those in the setting
of labeled transition systems; observe, moreover, that the original
labeled transition system already is minimal, demonstrating once more
that the minimisation of an equation system can be more effective than
minimising the original labeled transition system.

\section{Conclusions}
\label{Sect:Conclusions}

We presented a set of deduction rules for projecting the essential
information underlying Boolean equation systems onto so-called
\emph{structure graphs}. These graphs generalise the dependency graphs
of~\cite{Kei:06,KeirenWillemse2009a} which capture the dependencies of
closed equation systems in \emph{standard recursive form (SRF)}. We
showed that a minimisation of closed equation systems can be achieved
through a bisimulation minimisation of the underlying structure graphs,
and that this minimisation is sound: the minimised equation system
reflects and preserves the solution of the original equation system.
This generalises the results of~\cite{KeirenWillemse2009a}, in which
minimisation was possible only after bringing the closed equation system
into SRF.  The practical significance of minimisation of closed equation
systems in SRF was already addressed in~\cite{KeirenWillemse2009a}.

The work we have presented serves as a starting point for further
investigations. While we have not studied the problem of minimisation
of \emph{open} equation systems, extending our work in this direction
is not likely to raise problems of any significance, as our structure
graphs (with only small modifications) seems adequate for capturing
and reasoning about unbound variables.  Both from a theoretical
and a practical point of view, the study of weaker equivalences on
structure graphs is of importance. It is not fully clear whether the
\emph{idempotence-identifying bisimilarity} of~\cite{KeirenWillemse2009a},
which weakens some of the requirements of strong bisimilarity, carries
over to structure graphs without significant modifications. For equation
systems in SRF, this type of bisimulation solves, among others, the
idempotency problem that equations $\sigma X = X \wedge X$ and $\sigma X =
X$ are unrelated by strong bisimulation.  Furthermore, it would be very
interesting to study variations of stuttering equivalence in this context.

Finally, we consider a thorough understanding of the structure graphs
for BESs, and the associated notions of bisimilarity defined thereon, as
a first step towards defining similar-spirited notions in the setting of
\emph{parameterised Boolean equation systems}~\cite{GW:05b}. The latter
are high-level, symbolic descriptions of Boolean equation systems.
The advantage of such a theory would be that it would lead to elegant,
short proofs of various PBES manipulations that currently require lengthy
and tedious (transfinite) inductive proofs.

\paragraph{Acknowledgements.} We thank Jeroen Keiren for his valuable comments on a preliminary version.

\end{document}